\documentclass[12pt, tikz]{article}
\usepackage[margin=0.9in]{geometry} 

\usepackage{verbatim}

\usepackage{graphicx}
\usepackage{multirow} 
\usepackage{array,booktabs,arydshln,xcolor}

\usepackage{listings}
\usepackage{enumerate}
\usepackage{enumitem}
\usepackage{float}
\usepackage{tikz}
\usetikzlibrary{decorations.pathreplacing}
\usetikzlibrary {decorations.markings,fpu} 
\usetikzlibrary{arrows}
\usetikzlibrary{decorations.pathmorphing}
\usetikzlibrary{patterns,patterns.meta}
\usetikzlibrary{decorations.pathmorphing,
                calligraphy,
				patterns.meta}
				\usepackage{tkz-euclide}

\usepackage{caption}
\usepackage{mathtools}
\usepackage{hyperref}
\usepackage{fancyhdr}
\usepackage{relsize}
\usepackage{thmtools}
\usepackage{thm-restate}
\usepackage{amsfonts}
\usepackage{amsmath}
\usepackage{bbm}
\usepackage{bm}
\usepackage{adjustbox}
\usepackage{ifthen}
\usepackage[ruled,linesnumbered]{algorithm2e}
\usepackage{amsthm, amssymb} 

\definecolor{toc}{RGB}{13,55,174}	
\usepackage{hyperref}				
\hypersetup{
    colorlinks=true,
    citecolor=toc,
    filecolor=black,
    linkcolor=toc,
    urlcolor=toc
}

\allowdisplaybreaks 

\newtheorem{theorem}{Theorem}[section]
\newtheorem{lemma}{Lemma}[theorem]
\newtheorem{corollary}{Corollary}[theorem]

\newcommand{\dist}{\mathcal{D}}
\newcommand{\boxes}{\mathcal{B}}

\newcommand{\lp}{\left}
\newcommand{\rp}{\right}
\newcommand{\E}[2]{\mathbb{E}_{#1}\lp[#2 \rp]}
\renewcommand{\Pr}[2]{\textbf{Pr}_{#1}\lp[#2 \rp]}

\newcommand{\e}{\varepsilon}
\newcommand{\s}{\sigma}

\renewcommand{\vec}{\bm}
\newcommand\opt{\text{OPT}}

\newcommand\alg{\text{ALG}}

\newcommand{\ind}[1]{\mathbbm{1}{ \lp\{ #1 \rp\} }}

\newcommand{\scenarios}{\mathcal{S}}
\newcommand{\pb}{\mathcal{PB}}
\newcommand{\tree}{\mathcal{T}}
\newcommand{\pbText}{\textsc{Pandora's Box}}
\newcommand{\msscText}{\textsc{Min Sum Set Cover}}

\newcommand{\child}{\text{child}}
\newcommand{\poly}{\text{poly}}

\date{} 

\title{Weitzman's Rule for Pandora's Box with Correlations}

\author{
Evangelia Gergatsouli \\ UW-Madison \\ {\tt evagerg@cs.wisc.edu} \and 
Christos Tzamos \\ UW-Madison \& \\
University of Athens \\ {\tt tzamos@wisc.edu} 
}

\begin{document}

\maketitle

\begin{abstract}
    \pbText{}  is a central problem in decision making under uncertainty that can model various real life scenarios. In this problem we are given $n$ boxes, each with a fixed opening cost, and an unknown value drawn from a known distribution, only revealed if we pay the opening cost. Our goal is to find a strategy for opening boxes to minimize the sum of the value selected and the opening cost paid.

In this work we revisit \pbText{} when the value distributions are correlated, first studied in \cite{ChawGergTengTzamZhan2020}. We show that the optimal algorithm for the independent case, given by Weitzman's rule, directly works for the correlated case. In fact, it results in significantly improved approximation guarantees than the previous work. We also show how to implement the rule given only sample access to the correlated distribution of values. Specifically, we find that a number of samples that is polynomial in the number of boxes is sufficient for the algorithm to work.

\end{abstract}
\setcounter{page}{0}
\thispagestyle{empty}
\newpage

\section{Introduction}
In various minimization problems where uncertainty exists in the input, we are allowed to obtain information to remove this uncertainty by paying an extra price. Our goal is to sequentially decide which piece of information to acquire next, in order to minimize the sum of the search cost and the value of the option we chose. 

This family of problems is naturally modeled by \pbText{}, first formulated by Weitzman~\cite{Weit1979} in an economics setting. In this problem
where we are given $n$ boxes, each containing a value drawn from a known distribution and each having a fixed known \emph{opening cost}. We can only see the exact value realized in a box if we open it and pay the opening cost. Our goal is to minimize the sum of the value we select and the opening costs of the boxes we opened. 

In the original work of Weitzman, an optimal solution was proposed when the distributions on the values of the boxes were independent~\cite{Weit1979}. 
%
This algorithm was based on calculating a \emph{reservation value} ($\s$) for each box, and then 
choosing the box with the lowest reservation value to open at every step.
Given that independence is an unrealistic assumption in real life, \cite{ChawGergTengTzamZhan2020} first studied the problem where the distributions are correlated, and designed an algorithm giving a constant approximation guarantee. This algorithm is 
quite involved, it requires solving an LP to convert the \pbText{} instance to a \msscText{} one, and then solving this instance to obtain an ordering of opening the boxes. Finally, it reduces the problem of deciding when to stop to an online algorithm question corresponding to \textsc{Ski-Rental}.

 

\subsection{Our Contribution}
In this work we revisit \pbText{} with correlations, and provide \textbf{simpler}, \textbf{learnable} algorithms with \textbf{better approximation guarantees}, that directly \textbf{generalize} Weitzman's reservation values. More specifically, our results are the following. 
%
\begin{itemize}
    \item \textbf{Generalizing}: we first show how the original reservation values given by Weitzman~\cite{Weit1979} can be generalized to work in correlated distributions, thus allowing us to use a version of their initial greedy algorithm.
    \item \textbf{Better approximation}: we give two different variants of our main algorithm, that each uses different updates on the distribution $\dist$ after every step. 
    
    \begin{enumerate}
        \item \emph{Variant 1: partial updates}. We condition on the algorithm not having stopped yet.
        \item \emph{Variant 2: full updates}. We condition on the exact value $v$ revealed in the box opened.
    \end{enumerate}
    Both variants improve the approximation given by \cite{ChawGergTengTzamZhan2020} from $9.22$ to $4.428$ for Variant 1 and to $5.828$ for Variant 2.

    \item \textbf{Simplicity}: our algorithms are greedy and only rely on the generalized version of the reservation value, while the algorithms in previous work rely on solving a linear program, and reducing first to \msscText{} then to \textsc{Ski-Rental}, making them not straightforward to implement. A $9.22$ approximation was also given in \cite{GergTzam2022}, which followed the same approach but bypassed the need to reduce to \msscText{} by directly rounding the linear program via randomized rounding.
    
    \item \textbf{Learnability}: we show how given sample access to the correlated distribution $\dist$ we are able to still maintain the approximation guarantees. Specifically, for Variant 1 only $\poly(n, 1/\e, \log (1/\delta))$ samples are are enough to obtain $4.428+\e$ approximation with probability at least $1-\delta$. Variant 2 is however impossible to learn.
\end{itemize}

Our analysis is enabled by drawing similarities from \pbText{} to \msscText{}, which corresponds to the special case of when the values inside the boxes are $0$ or $\infty$. For \msscText{} a simple greedy algorithm was shown to achieve the optimal $4$-approximation~\cite{FeigUrieLovaTeta2002}. Surprisingly, Weitzman's algorithm can be seen as a direct generalization of that algorithm. Our analysis follows the histogram method introduced in~\cite{FeigUrieLovaTeta2002}, for bounding the approximation ratio. However, we significantly generalize it to handle values in the boxes and work with tree-histograms required to handle the case with full-updates. 



\subsection{Related Work}
Since Weitzman's initial work~\cite{Weit1979} on \pbText{} there has been a renewed interest in studying this problem in various settings. Specifically~\cite{Dova2018,BeyhKlei2019} study \pbText{} when we can select a box without paying for it (
non-obligatory inspection), in~\cite{BoodFuscLazoLeon2020} there are tree or line constraints on the order in which the boxes can be opened. In~\cite{ChawGergTengTzamZhan2020, ChawGergMcmaTzam2021} the distributions on the values inside the boxes are correlated and the goal is to minimize the search and value cost, while finally in~\cite{BechDughPate2022} the task of searching over boxes is delegated by an agent to a principal, while the agent makes the final choice. The recent work of Chawla et al.~\cite{ChawGergTengTzamZhan2020} is the first one that explores the correlated distributions variant and gives the first approximation guarantees.

This problem can be seen as being part of the ``price of information" literature~\cite{CharFagiGuruKleiRaghSaha2002, GuptKuma2001, ChenJavdKarbBagnSrinKrau2015,
ChenHassKarbKrau2015}, where we can remove part of the uncertainty of the problem at hand by paying a price. In this line of work, more recent papers study the structure of approximately optimal rules for combinatorial problems~\cite{GoelGuhaMuna2006, GuptNaga2013,
AdamSvirWard2016, GuptNagaSing2016, GuptNagaSing2017, Sing2018, GuptJianSing2019}.  
 

For the special case of \msscText{}, since the original work of~\cite{FeigUrieLovaTeta2002}, there has been many follow-ups and generalizations where every set has a requirement of how many elements contained in it we need to choose~\cite{AzaGamzIftaYinr2009, BansGuptRavi2010, AzarGamz2010,SkutWill2011,ImSvirZwaa2012}. 

\section{Preliminaries}
In \pbText{} ($\pb$) we are given a set of $n$ boxes $\boxes$, each with a known opening cost $c_b\in\mathbb{R}^+$, and a 
distribution $\dist$ over a vector of unknown values $\vec v = (v_1,\ldots,v_n) \in  \mathbb{R}_+^d$ inside the boxes. Each box $b\in \boxes$, once it is opened, reveals the value $v_{b}$. The algorithm can open boxes sequentially, by paying the opening cost each time, and observe the value instantiated inside the box. The goal of the algorithm is to choose a box of small value, while spending as little cost as possible ``opening" boxes. Formally, denoting by $\mathcal{O} \subseteq \boxes$ the set of opened boxes, we want to minimize
\[\E{v\sim \dist}{\sum_{b\in \mathcal{O}}c_b + \min_{b\in \mathcal{O}} v_b }.\]

A \emph{strategy} for \pbText{} is an algorithm that in every step decides which is the next box to open and when to stop. A strategy can pick any open box to select at any time. To model this, we assume wlog that after a box is opened the opening cost becomes $0$, allowing us to select the value without opening it again. In its full generality, a strategy can make decisions based on every box opened and value seen so far. We call this the \emph{Fully-Adaptive} (FA) strategy.

\paragraph{Different Benchmarks.} As it was initially observed
in~\cite{ChawGergTengTzamZhan2020}, optimizing over the class of fully-adaptive strategies is intractable, therefore we consider the simpler benchmark of \emph{partially-adaptive} (PA) strategies. In this case, the algorithm has to fix the opening order of the boxes, while the stopping rule can arbitrarily depend on the values revealed.

\subsection{Weitzman's Algorithm}
When the distributions of values in the boxes are independent, Weitzman~\cite{Weit1979} described a greedy algorithm that is also the optimal strategy. 
In this algorithm, we first calculate an index for every box $b$, called \emph{reservation value} $\s_b$, defined as the value that satisfies the following equation
\begin{equation}\label{eq:reservation_value}
    \E{\vec v \sim \mathcal{D}}{(\s_b-v_b)^+} = c_b.
\end{equation}

Then, the boxes are ordered by increasing $\sigma_b$ and opened until the minimum value revealed is less than the next box in the order. Observe that this is a \emph{partially-adaptive} strategy.


\section{Competing with the Partially-Adaptive}\label{sec:main_algo}
We begin by showing how Weitzman's algorithm can be extended to correlated distributions. Our algorithm calculates a reservation value $\s$ for every box at each step, and opens the box $b\in \boxes$ with the minimum $\s_b$. We stop if the value is less than the reservation value calculated, and proceed in making this box \emph{free}; we can re-open this for no cost, to obtain the value just realized at any later point. The formal statement is shown in Algorithm~\ref{alg:weitzman_main}. 

We give two different variants based on the type of update we do after every step on the distribution $\dist$. In the case of partial updates, we only condition on $V_b>\sigma_b$, which is equivalent to the algorithm not having stopped. On the other hand, for full updates we condition on the exact value that was instantiated in the box opened. Theorem~\ref{thm:main_vs_pa} gives the approximation guarantees for both versions of this algorithm.

\begin{algorithm}[ht]
		\KwIn{Boxes with costs $c_i\in \mathbb{R}$, distribution over scenarios $\mathcal{D}$.}
       An unknown vector of values $v \sim \mathcal{D}$ is drawn\\
        \Repeat{termination}{
       Calculate $\sigma_b$ for each box $b \in \boxes$ by solving:
      $$
      \E{\vec v \sim \mathcal{D}}{(\s_b-v_b)^+} = c_b.
      $$\\
     Open box $b = \text{argmin}_{b\in \boxes}\sigma_b$\\
      Stop if the value the observed $V_b = v_b \le \sigma_b$\\
     $c_b\leftarrow 0$ \tcp{Box is always open now}
         Update the prior distribution
        \begin{itemize}
            \item[-] \textbf{Variant 1}: $\mathcal{D} \leftarrow \mathcal{D}|_{V_b > \sigma_b}$ (partial updates)
            \item[-] \textbf{Variant 2}: $\mathcal{D} \leftarrow \mathcal{D}|_{V_b = v_b} $ (full updates)
        \end{itemize} 
        }
		\caption{Weitzman's algorithm, for correlated $\dist$. }\label{alg:weitzman_main}
	\end{algorithm}

\begin{restatable}{theorem}{mainThm}\label{thm:main_vs_pa}
	Algorithm~\ref{alg:weitzman_main} is a $4.428$-approximation for Variant 1 and $5.828$-approximation for Variant 2 of \pbText{}
  against the partially-adaptive optimal.
\end{restatable}

\begin{proof}
We seperately show the two components of this theorem in Theorems~~\ref{thm:main_vs_pa_partial} and \ref{thm:main_vs_pa_equal}.
\end{proof}

Observe that for independent distributions this algorithm is exactly the same as Weitzman's \cite{Weit1979}, since the product prior $\dist$ remains the same, regardless of the values realized. Therefore, the calculation of the reservation values does not change in every round, and suffices to calculate them only once at the beginning.

\paragraph{Scenarios} To proceed with the analysis of Theorem~\ref{thm:main_vs_pa}, we assume that $\dist$ is supported on a collection of $m$ vectors, $(\vec v^s)_{s \in \scenarios}$, which we call scenarios, and sometimes abuse notation to say that a scenario is sampled from the distribution $\dist$. 
We assume that all scenarios have equal probability. The general case with unequal probabilities follows by creating more copies of the higher probability scenarios until the distribution is uniform.

A scenario is \emph{covered} when the algorithm decides to stop and choose a value from the opened boxes. For a specific scenario $s\in\scenarios$ we denote by $c(s)$ the total opening cost paid by an algorithm before this scenario is covered and by $v(s)$ the value chosen for this scenario.


%
\paragraph{Reservation Values}
To analyze Theorem~\ref{thm:main_vs_pa}, we introduce a new way of defining the reservation values of the boxes that is equivalent to~\eqref{eq:reservation_value}.
For a box $b$, we have that

\[\s_b = \min_{A \subseteq \scenarios} \frac{c_b+ \sum_{s\in A} \Pr{\dist}{s} v^s_b}{\sum_{s\in A} \Pr{\dist}{s}} \]

The equivalence to~\eqref{eq:reservation_value}, follows since $\s_b$ is defined as the root of the expression
\begin{align*}
      \E{s \sim \mathcal{D}}{(\s_b-v^s_b)^+} & - c_b = \sum_{s\in \scenarios} \Pr{\dist}{s} (\s_b - v^s_b)^+ - c_b \\
      &= \max_{A \subseteq \scenarios} \sum_{s\in A} \Pr{\dist}{s} (\s_b - v^s_b)  - c_b.
\end{align*}

Thus, $\s_b$ is also the root of
\begin{align*}
      \max_{A \subseteq \scenarios}& \frac { \sum_{s\in A} \Pr{\dist}{s} (\s_b - v^s_b)  - c_b } {\sum_{s\in A} \Pr{\dist}{s}} 
      = \s_b - \min_{A \subseteq \scenarios} \frac { c_b + \sum_{s\in A} \Pr{\dist}{s} v^s_b } {\sum_{s\in A} \Pr{\dist}{s}}.
\end{align*}

This, gives our formula for computing $\s_b$, which we can further simplify using our assumption that all scenarios have equal probability. In this case, $ \Pr{\dist}{s}= 1/|\scenarios|$ which implies that
\begin{equation}\label{eq:equivalent_reserv}
\s_b = \min_{A \subseteq \scenarios} \frac{c_b|\scenarios|+\sum_{s\in A}v_s}{|A|}.
\end{equation}

\subsection{Conditioning on $V_b>\sigma_b$}\label{subsec:proof_algo}

We start by describing the simpler variant of our algorithm where after opening each box we update the distribution by conditioning on the event $V_b > \s_b$. This algorithm is \emph{partially adaptive}, since the order for each scenario does not depend on the actual value that is realized every time. At every step the algorithm will either stop or continue opening boxes conditioned on the event ``We have not stopped yet" which does not differentiate among the surviving scenarios.

\begin{restatable}{theorem}{mainThm}\label{thm:main_vs_pa_partial}
Algorithm~\ref{alg:weitzman_main} is a $4.428$-approximation  
  for \pbText{} 
  against the partially-adaptive optimal, when conditioning on $V_b > \sigma_b$.
\end{restatable}

In this section we show a simpler proof for Theorem~\ref{thm:main_vs_pa_partial} that gives a $3+2\sqrt{2} \approx 5.828$-approximation. The full proof for the $4.428$-approximation is given in section~\ref{apn:main_algo} of the Appendix. Using the equivalent definition of the reservation value (Equation~\eqref{eq:equivalent_reserv}) we can rewrite Algorithm~\ref{alg:weitzman_main} as follows.\\

%

\begin{algorithm}[H]
	\KwIn{Boxes with costs $c_i\in \mathbb{R}$, set of scenarios $\scenarios$.}
            $t\leftarrow 0$\\
		$R_0 \leftarrow \mathcal{S}$ the set of scenarios still uncovered\\
			\While{$R_t\neq \emptyset$}{
			 Let $\s_t \leftarrow \text{min}_{b\in \boxes, A\subseteq R_t} \frac{c_b|R_t| + \sum_{s\in A}v^s_{b}}{|A|}$\\
      Let $b_t$ and $A_t$ be the box and the set of scenarios that achieve the minimum\\
				Open box $b_t$ and pay $c_{b_t}$\\
    Stop and choose the value at box $b_t$ if it is less than $\s_t$: this holds \textbf{iff} $s \in A_t$\\
               Set $c_{b_t} \leftarrow 0$\\
				$R_t \leftarrow R_t \setminus A_t$\\
				$t\leftarrow t+1$
    }
			\caption{Weitzman's rule for Partial Updates}\label{alg:general_v}
	\end{algorithm}

We first start by giving a bound on the cost of the algorithm. 
The cost can be broken down into opening cost plus the value obtained. 
Since at any time $t$, all remaining scenarios $R_t$ pay the opening cost $c_{b_t}$, we have that the total opening cost is $$\sum_t c_{b_t} |R_t|.$$ Moreover,
the chosen value is given as $$\sum_t \sum_{s \in A_t} v^s_{b_t}. $$
Overall, we have that
\begin{align*}
\alg & = \sum_t \left( c_{b_t} |R_t| + \sum_{s \in A_t} v^s_{b_t} \right) 
= \sum_t |A_t| \frac {  c_{b_t} |R_t| + \sum_{s \in A_t} v^s_{b_t} } { |A_t|} = \sum_t |A_t| \s_t. 
\end{align*}

Defining $\s_s$ to be the reservation value of scenario $s$ at the time it is covered, i.e. when $s \in A_t$, we get $\alg = \sum_{s \in \scenarios} \s_s$. We follow a \emph{histogram analysis} similar to the proof of Theorem~4 in \cite{FeigUrieLovaTeta2004} for \msscText{} and construct the following histograms.
\begin{itemize}
	\item The $\opt_o$ histogram: put the scenarios on the x-axis on increasing
opening cost order $c_s^\opt$ according to $\opt$, the height of each scenario is the opening cost it paid. 
	\item The $\opt_v$ histogram: put the scenarios on the x-axis on increasing
covering value order $v_s^\opt$ according to $\opt$, the height of each scenario is the value with which it
was covered.
	\item The $\alg$ histogram: put scenarios on the x-axis in the order the algorithm covers them. The height of each scenario is $\s_s$. Observe that the area of the $\alg$ histogram is exactly the cost of the algorithm.
\end{itemize}


\begin{proof}[Proof of Theorem~\ref{thm:main_vs_pa_partial}]
Initially, observe that the algorithm will eventually stop; every time we open a box we cover at least one scenario (since line 3 is cannot be $\infty$ while scenarios are left uncovered).

	To show the approximation factor, we scale the histograms as follows; $\opt_o$ scale horizontally by $1/\alpha_o$ and vertically by $1/(\beta \cdot \gamma)$, and $\opt_v$ scale by $1/\alpha_v$ horizontally, for some constants $ \alpha_o, \alpha_v \in (0,1)$ to be determined later\footnote{Scaling horizontally means that we duplicate every scenario and scaling vertically we just multiply the height at every point by the scale factor.}.
  We align the $\alg$ histogram with $\opt_v$ and $\opt_o$ so that all of them have the same right-hand side. Observe that the optimal opening cost is the area below the histogram $\opt_o$ and has increased by $\beta\cdot \gamma \cdot \alpha_o$, and similarly the area below $\opt_v$ has increased by $\alpha_v$ as a result of the scaling.

To conclude the proof it suffices to show that any point in the $\alg$ histogram is inside the sum of the rescaled $\opt_v$ and $\opt_o$ histograms. Consider any point $p$ in the $\alg$ histogram, and let $s$ be its corresponding scenario and $t$ be the time this scenario is covered. We have that the height of the $\alg$ histogram is
\begin{equation}\label{eq:sigma_p}
	\sigma_s = \frac{c_{b_t} |R_t| + \sum_{s\in A_t}v^s_{b_t}}{|A_t|} 
		\leq \frac{c_b|R_t| + \sum_{s\in A}v^s_{b}}{|A|}
\end{equation}
where the last inequality holds for all $A\subseteq R_t$ and any $b\in \boxes$.

  Denote by $c^*$ the opening cost such that $\gamma|R_t|$ of the scenarios in $R_t$ have opening cost less than $c^*$, and by $R_{\text{low}} = \{s\in R_t: c_s^\opt \leq c^*\}$ the set of these scenarios. Similarly denote by $v^*$ the value of scenarios in $R_{\text{low}}$ such that $\beta |R_{\text{low}}|$ of the scenarios have value less than $v^*$ and by $L = \{s\in R_{\text{low}}: v_s^\opt\leq v^*\}$ these scenarios. This split is shown in Figure~\ref{fig:t_v_split}, and the constants $\beta, \gamma \in (0,1)$ will be determined at the end of the proof.

		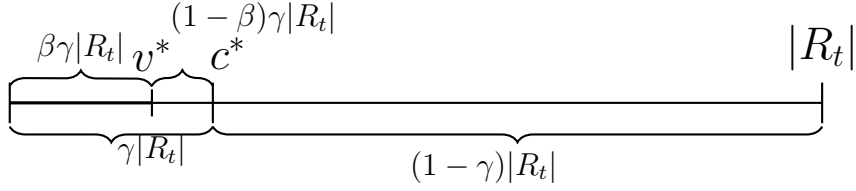
\begin{figure*}[ht]
				\centering
    \begin{tikzpicture}[scale=0.9]

\pgfmathtruncatemacro{\labelsY}{0.2};
\pgfmathtruncatemacro{\split}{3};

\draw[-,very thick, black] (0,0) -- (\split*0.7,0);
\draw[-,thick] (\split*0.7,0) -- (12,0);

\draw [decorate,black,thick,decoration={brace,amplitude=6pt}] (0,0.25) -- (\split*0.7,0.24) node[]{};
\draw [decorate,thick,decoration={brace,amplitude=6pt, mirror}] (0,-0.35) -- (\split,-0.35) node[]{};
\draw [decorate,thick,decoration={brace,amplitude=6pt, mirror}] (\split,-0.35)--(12,-0.35) node[]{};
\draw [decorate,thick,decoration={brace,amplitude=6pt}] (\split*0.7,0.25) -- (\split,0.25) node[]{};

\draw[-, thick] (0,0.3)-- (0,-0.3) node[]{};
\draw[-, thick] (12,0.3)-- (12,-0.3) node[label={[yshift=0.4cm]{ {\Large $|R_t|$}}}]{};
\draw[-, thick] (\split*0.7,0.2)-- (\split*0.7,-0.2) node[label={[yshift=0.3cm] {\Large $v^*$} }]{};
\draw[-, thick] (\split,0.3)-- (\split,-0.3) node[label={[xshift=0.2cm,yshift=0.42cm]{\Large $c^*$}}]{};

\draw (\split*0.7,-1.25) node[label={$\gamma|R_t|$}]{};
\draw (7,-1.5) node[label={$(1-\gamma)|R_t|$}]{};
\draw (\split*0.35,0.25) node[label={\textcolor{black}{$\beta \gamma|R_t|$}}]{};
\draw (1.2*\split,0.7) node[label={$(1-\beta)\gamma|R_t|$}]{};

\end{tikzpicture}
				\caption{Split of scenarios in $R_t$.}
    \label{fig:t_v_split}
\end{figure*}

Let $B_L$ be the set of boxes that the optimal solution uses to cover the scenarios in $L$. Let $L_b \subseteq L \subseteq R_t$ be the subset of scenarios in $L$ that choose the value at box $b$ in $\opt$. Using inequality~\eqref{eq:sigma_p} with $b\in B_L$ and $A = L_b$, we obtain $\sigma_s |L_b| \leq c_b |R_t| + \sum_{s \in L_b} v_s^\opt $, and by summing up the inequalities for all $b\in B_L$ we get 
\begin{align}\label{eq:sigma_UB}
\sigma_s &\leq \frac{|R_t| \sum_{b\in B_L} c_b + \sum_{s\in L} v_s^\opt}{|L|}\\
		&\leq \frac{|R_t| c^* + \sum_{s\in L} v_s^\opt}{|L|} 
		\leq \frac{c^*}{\beta \cdot \gamma} + \frac{\sum_{s\in L} v_s^\opt}{|L|}
\end{align}
where for the second inequality we used that the cost for covering the scenarios in $L$ is at most $c^*$ by construction, and in the last inequality that $|L| = |R_t|/(\beta \cdot \gamma)$. We consider each term above separately, to show that the point $p$ is within the histograms.

\paragraph{Bounding the opening cost.}
By the construction of $c^*$, the point in the $\opt_o$ histogram that has cost at least $c^*$ is at distance at least $(1-\gamma)|R_t|$ from the right hand side. This means that in the rescaled histogram, 
the point that has cost at least $c^*/(\beta\cdot \gamma)$ is at distance at least $(1-\gamma)|R_t| / \alpha_o$ from the right hand side.

On the other hand, in the $\alg$ histogram the distance of $p$ from the right edge of the histogram is at most $|R_t|$, therefore for the point $p$ to be inside the $\opt_o$ histogram we require 
\begin{equation}\label{eq:at}
		\alpha_o\leq 1-\gamma.
		\end{equation}

\paragraph{Bounding the values cost.}
 By the construction of $v^*$, the point in the $\opt_v$ histogram that has value $v^*$ is at distance at least $|R_t|(1-\beta )\gamma$ from the right hand side. This means that in the rescaled histogram, the point that has value at least $v^*$ is at distance at least $(1-\beta )\gamma|R_t| / \alpha_v$ from the right hand side.

On the other hand, in the $\alg$ histogram the distance of $p$ from the right edge of the histogram is at most $|R_t|$, therefore for the point $p$ to be inside the $\opt_o$ histogram we require 
	\begin{equation}\label{eq:av} 
			\alpha_v\leq (1-\beta )\gamma.
		\end{equation}

	We optimize the constants $\alpha_o, \alpha_v, \beta ,\gamma $ by ensuring that
	inequalities~\eqref{eq:at} and \eqref{eq:av} hold. We set $\alpha_o= 1-\gamma$ and 
	$\alpha_v = (1-\beta)\gamma$, and obtain that  
	$\alg \le \opt_o/(\beta\cdot \gamma \cdot (1-\gamma)) + \opt_v/((1-\beta)\gamma) $. 
	Requiring these to be equal we get $\beta=1/(2-\gamma)$, which is minimized for $\beta=1/\sqrt{2}$ and $
	\gamma=2-\sqrt{2}$ for a value of $3+2\sqrt{2}$.

\end{proof}

\subsection{Conditioning on $V_b=v$}\label{subsec:equal}
In this section we switch gears to our second variant of Algorithm~\ref{alg:weitzman_main}, where in each step we update the prior $\dist$ conditioning on the event $V_b = v$. We state our result in Theorem~\ref{thm:main_vs_pa_equal}. In this case, the conditioning on $\dist$ implies that the algorithm at every step removes the scenarios that are \emph{inconsistent} with the value realized. 

\begin{restatable}{theorem}{mainThmVarEqual}\label{thm:main_vs_pa_equal}
	Algorithm~\ref{alg:weitzman_main} is a $3+2\sqrt{2} \approx 5.828$-approximation 
  for \pbText{} against the partially-adaptive optimal, when conditioning on $V_b  = v$.
\end{restatable}

The main challenge was that the algorithm's solution is now a tree with respect to scenarios instead of a line as in the case of $\dist|_{V_b>\sigma_b}$. Specifically, in the $D|_{V_b>\sigma_b}$ variant at every step all scenarios that had $V_b \leq \sigma_b$ were covered and removed from consideration. However in the $D|_{V_b=v}$ variant the remaining scenarios are split into different cases, based on the realization of $V$, as shown in the example of Figure~\ref{fig:tree_solution_equal}.

\begin{figure}[ht]
    \centering
    \begin{tikzpicture}

\tikzset{My Style/.style={shape=circle, draw=black, minimum size=30pt}}

\pgfmathsetmacro{\spacing}{2.5}
\pgfmathsetmacro{\offsetx}{1}

\node [My Style,label={[yshift=-0.6cm,xshift=1.1cm]Open $b_2$}] (root) at (0,0) {$s_1, s_2, s_3$};

\node[My Style,label={[yshift=-0.6cm,xshift=-1cm]Stop}] (left) at (-\offsetx,-\spacing) {$s_1$};
\node[My Style,label={[yshift=-0.6cm,xshift=1.1cm]Open $b_1$}] (right) at (\offsetx,-\spacing) {$s_2, s_3$};

\draw[-] (root) -- node[left] {$V=2$} ++ (left);
\draw[-] (root) -- node[right] {$V=5$} ++ (right);

\node [My Style,label={[yshift=-0.6cm,xshift=-1cm]Stop}] (rightLeft) at (-0.3*\offsetx,-2*\spacing) {$s_3$};

\node [My Style,label={[yshift=-0.6cm,xshift=1cm]Stop}] (rightRight) at (2*\offsetx,-2*\spacing) {$s_2$};

\draw[-] (right) -- node[left] {$V=2$} ++ (rightLeft);
\draw[-] (right) -- node[right] {$V=1$} ++ (rightRight);

\end{tikzpicture}
    \caption{Algorithm's solution when conditioning on $V=v$, for an instance with scenarios $\scenarios=\{s_1, s_2, s_3\}$, and boxes $\boxes = \{b_1, b_2\}$. The nodes contain the consistent scenarios at each step, and the values $V$ are revealed once we open the corresponding box.}
    \label{fig:tree_solution_equal}
\end{figure}
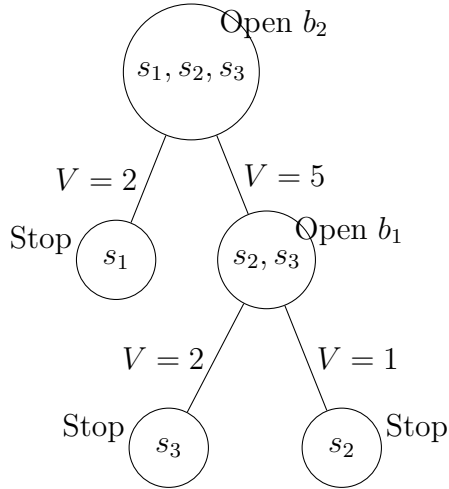

This results into the ALG histogram not being well defined, since there is no unique order of covering the scenarios. We overcome this by generalizing the histogram approach to trees.

\begin{proof}[Proof of Theorem~\ref{thm:main_vs_pa_equal}]
The proof follows similar steps to that of Theorem~\ref{thm:main_vs_pa_partial}, thus we only highlight the differences. The algorithm  is presented below, the only change is line 5 where we remove the inconsistent with the value revealed scenarios, which also leads to our solution branching out for different scenarios and forming a tree.

 \begin{algorithm}[ht]
				 \KwIn{Boxes with costs $c_i\in \mathbb{R}$, set of scenarios $\scenarios$.}
        Define a root node $u$ corresponding to the set $\mathcal{S}$\\
		 $R_u \leftarrow \mathcal{S}$ the set of scenarios still uncovered\\
			\While{$R_u\neq \emptyset$}{
				Let $\s_u \leftarrow \text{min}_{b\in \boxes, A\subseteq R_u} \frac{c_b|R_u| + \sum_{s\in A}v^s_{b}}{|A|}$\\
      Let $b_u$ and $A_u$ be the box and the set of scenarios that achieve the minimum\\
			Open box $b_u$ paying $c_{b_u}$ and observe value $v$\\
    Stop and choose the value at box $b_u$ if it is less than $\s_u$: this holds \textbf{iff} $s \in A_u$\\
               Set $c_{b_u} \leftarrow 0$\\
             Let $u'$ be a vertex corresponding to the set of consistent scenarios with
    $R_{u'} \triangleq R_u \setminus \lp(A_u \cup \{s\in R_u: v_{b_u}^s\neq v\}\rp)$ \tcp{Remove inconsistent scenarios}
			 	Set $u\leftarrow u'$
		}
			\caption{Weitzman's rule for Full Updates    }\label{alg:general_v_equal}
	\end{algorithm}

\paragraph{Bounding the opening cost}
Consider the tree $\tree$ of $\alg$ where at every node $u$ a set $A_u$ of scenarios is covered. We associate this  tree with node weights, where at every node $u$, we assign $|A_u|$ weights $(\s_u,...,\s_u)$. Denote, the weighted tree by $\tree_\alg$. As before, the total cost of $\alg$ is equal to the sum of the weights of the tree.

We now consider two alternative ways of assigning weights to the the nodes, forming trees $\tree_{\opt_o}$, $\tree_{\opt_v}$ using the following process. 
\begin{itemize}
    \item $\tree_{\opt_o}.$ At every node $u$ we create a vector of weights $\vec w^{\opt_o}_u = (c^\opt_{s})_{s \in A_u}$ where each $c^\opt_{s}$ is the opening cost that scenario $s \in A_u$ has in the optimal solution.
    \item $\tree_{\opt_v}.$ At every node $u$ we create a vector of weights $\vec w^{\opt_v}_u = (v^\opt_{s})_{s \in A_u}$ where each $v^\opt_{s}$ is the value the optimal uses to cover scenario $s \in A_u$.
\end{itemize}
We denote by $\text{cost}(\tree_\alg)$ the sum of all weights in every node of the tree $\tree$.
We have that $\text{cost}(\tree)$ is equal to the total cost of $\alg$, while 
$\text{cost}(\tree_{\opt_o})$ and $\text{cost}(\tree_{\opt_v})$ is equal to the optimal opening cost $\opt_o$ and optimal value $\opt_v$ respectively. Intuitively, the weighted trees correspond to the histograms in the previous analysis of Theorem~\ref{thm:main_vs_pa_partial}.

We want to relate the cost of $\alg$, to that of $\tree_{\opt_o}$ and $\tree_{\opt_v}$. To do this, we define an operation similar to histogram scaling, which replaces the weights of every node $u$ in a tree with the top $\rho$-percentile of the weights in the subtree rooted at $u$.  As the following lemma shows, this changes the cost of a tree by a bounded multiplicative factor.

\begin{restatable}{lemma}{genHist}\label{lem:general_histogram}
Let $\tree$ be a tree with a vector of weights $\vec w_u$ at each node $u\in \tree$, and let ${ \tree^{(\rho)} }$ be the tree we get when we substitute the weights of every node with the top $\rho$-percentile of all the weights in the subtree of $\tree$ rooted at $u$. Then%
\[\rho \cdot \text{cost}(\tree^{(\rho)}) \leq cost (\tree).\]
\end{restatable}

We defer the proof of Lemma~\ref{lem:general_histogram} to Section~\ref{apn:equal} of the Appendix.
To complete the proof of Theorem~\ref{thm:main_vs_pa_equal}, and bound $\text{cost}(\tree_\alg)$, we show as before that the weights at every node $u$, are bounded by the weights of $\tree^{(1-\gamma)}_{\opt_o}$ scaled by $\frac 1 {\beta \gamma}$ plus the weights of $\tree^{((1-\beta)\gamma)}_{\opt_v}$, for the constants $\beta, \gamma \in (0,1)$ chosen in the proof of Theorem~\ref{thm:main_vs_pa_partial}. This implies that
\begin{align*}
\text{cost}(\tree_{\opt_o})
\le &
\frac 1 {\beta \gamma} \text{cost}(\tree^{(1-\gamma)}_{\opt_o})
+
\text{cost}(\tree^{((1-\beta)\gamma)}_{\opt_v})\\
\le&
\frac 1 {\beta \gamma (1\text{-}\gamma)} \text{cost}(\tree_{\opt_o})
+
\frac 1 {(1\text{-}\beta)\gamma}
\text{cost}(\tree_{\opt_v})
\end{align*}
which gives $\alg \le 5.828\, \opt $ for the choice of $\beta$ and $\gamma$.
The details of the proof are similar to the one of Theorem~\ref{thm:main_vs_pa}, and are deferred to section~\ref{apn:equal} of the Appendix.

\end{proof}

\subsection{Lower Bound}

To show that our algorithm is almost tight, we observe that the lower bound of Min Sum Set Cover presented in~\cite{FeigUrieLovaTeta2004} also applies to \pbText{}. In \msscText{} we are given $n$ elements $e_i$, and $m$ sets $s_j$ where each $s_j \subseteq [n]$. We say a set $s_j$ \emph{covers} an element $e_i$  if $e_i \in s_j$ . The goal is to select elements in order to minimize the sum of the \emph{covering times} of all the sets, where \emph{covering time} of a set is he first time an element $e_i\in s_j$ is chosen. 

In \cite{FeigUrieLovaTeta2004} the authors show that \msscText{} cannot be approximated better than $4-\e$ even in the special case where every set contains the same number of elements\footnote{Equivalently forms a uniform hypergraph, where sets are hyperedges, and elements are vertices.}. We restate the theorem below.

\begin{theorem}[Theorem 13 of \cite{FeigUrieLovaTeta2004}]\label{thm:hardness_feige}
    For every $\e> 0$, it is NP-hard to approximate min sum set cover within a
ratio of $4 - \e$ on uniform hypergraphs.
\end{theorem}

Our main observation is that \msscText{} is a special case of \pbText{}. When the boxes all have the same opening cost $c_b=1$ and the values inside are $v^b_s \in \{0,\infty\}$, we are required to find a $0$ for each scenario; equivalent to \emph{covering} a scenario. The optimal solution of \msscText{} is an algorithm that selects elements one by one, and stops whenever all the sets are covered. This is exactly the partially adaptive optimal we defined for \pbText{}. The theorem restated above results in the following Corollary.

\begin{corollary}
    It is NP-Hard to approximate Pandora's Box against the partially-adaptive better than $4-\e$.
\end{corollary}


\section{Learning from Samples}\label{sec:learning}
In this section we show that our algorithm also works when we are only given sample access to the correlated distribution $\dist$.

We will mainly focus on the first variant with partial updates $\dist|_{V>v}$.
The second variant with full Bayesian updates $\dist|_{V=v}$ is learnable requires full knowledge of the underlying distribution and can only work with sample access if one can learn the full distribution. To see this consider for example an instance where the values are drawn uniformly from $[0,1]^d$. No matter how many samples one draws, it is impossible to know the conditional distribution $\dist|_{V=v}$ after opening the first box for a fresh samples $v$, and the Bayesian update is not well defined.

Variant 1 does not have this problem and can be learned from samples if the costs of the boxes are polynomially bounded by $n$, i.e. if there is a constant $c > 0$ such that for all $b \in \boxes$, $c_b \in [1,n^c]$.  If the weights are unbounded, it is impossible to get a good approximation with few samples. To see this consider the following instance.
Box 1 has cost $1/H \rightarrow 0$, while every other box has cost $H$ for a very large $H>0$. Now consider a distribution
where with probability $1-\frac 1 H \rightarrow 1$, the value in the first box is $0$, and with probability $1/H$ is $+\infty$. In this case, with a small number of samples we never observe any scenario where $v_1 \neq 0$ and believe the overall cost is near $0$. However, the true cost is at least $H \cdot 1/H \ge$ and is determined by how the order of boxes is chosen when the scenario has $v_1 \neq 0$. Without any such samples it is impossible to pick a good order.

Therefore, we proceed to analyze Variant 1 with $\dist|_{V>\sigma}$ in the case when the box costs are similar. We show that polynomial, in the number of boxes, samples suffice to obtain an approximately-optimal algorithm, as we formally state in the following theorem. We present the case where all boxes have cost 1 but the case where the costs are polynomially bounded easily follows.

\begin{theorem}\label{thm:learning}
Consider an instance of Pandora's Box with opening costs equal to 1. 
For any given parameters $\e, \delta>0$, 
using $m = poly(n, 1/\e, \log(1/\delta))$ 
samples from $\dist$, Algorithm~\ref{alg:weitzman_main} (Variant 1) obtains a $4.428 + \e$ approximation policy against the partially-adaptive optimal, with probability at least $1-\delta$.
\end{theorem}

To prove the theorem, we first note that  variant 1 of Algorithm~\ref{alg:weitzman_main} takes a surprisingly simple form, which we call a threshold policy.
It can be described by a permutation $\pi$ of visiting the boxes and a vector of thresholds $\bm{\tau}$ that indicate when to stop. The threshold for every box corresponds to the reservation value the first time the box is opened. To analyze the sample complexity of Algorithm~\ref{alg:weitzman_main}, we study a broader class of algorithms parameterized by a permutation and vector of thresholds given in Algorithm~\ref{alg:learning}.



\begin{algorithm}[ht]
\KwIn{Set of boxes, permutation $\pi$, vector of thresholds $\bm{\tau}\in \mathbb{R}^n$}
best $\leftarrow \infty$\\
\ForEach{$i \in [n]$}{
\uIf{best $>\tau_i$}{
Open box $\pi_i$, see value $v_i$\\
 best $\leftarrow \min(\text{best}, v_i)$
}
\uElse{
Accept best
}
}
\caption{General format of \pbText{} algorithm.}\label{alg:learning}
\end{algorithm}
Our goal now is to show that polynomially many samples from the distribution $\dist$ suffice to learn good parameters for Algorithm~\ref{alg:learning}. We first show a Lemma  that bounds the cost of the algorithm calculated in the empirical $\hat{\dist}$ instead of the original $\dist$ (Lemma~\ref{lem:closeness}), and a Lemma~\ref{lem:capping} that shows how capping the reservation values by $n/\e$ can also be done with negligible cost.

\begin{restatable}{lemma}{closeness}\label{lem:closeness}
Let $\e, \delta > 0$ and let $\dist'$ be the empirical distribution  obtained from $\poly(n, 1/\e,\\ \log(1/\delta))$ samples from $\dist$. Then, with probability $1-\delta$, it holds that 
$$\left| \E{\hat D}{\alg(\pi,\tau) - \min_{b\in \boxes} v_b} - \E{D}{\alg(\pi,\tau) - \min_{b\in \boxes} v_b} \right| \le \e$$
for any permutation $\pi$ and any 
vector of thresholds $\vec v \in \lp[0, \frac n \e\rp]^n$
\end{restatable}

We defer the proof of Lemma~\ref{lem:closeness} the section~\ref{apn:learning} of the Appendix.

\begin{lemma}\label{lem:capping}
Let $\dist$ be any distribution of values. Let $\e > 0$ and consider a permutation $\pi$ and thresholds $\vec \tau$. Moreover, let $\tau'$ be the thresholds capped to $n/\e$, i.e. setting $\tau'_b = \min \{ \tau_b, n/\e \} $ for all boxes $b$. Then,
$$\E{v \sim D}{\alg(\pi,\tau')} \le (1+\e) \E{v \sim D}{\alg(\pi,\tau)}.$$
\end{lemma}

\begin{proof}

We compare the expected cost of $\alg$ with the original thresholds and the transformed one $\alg'$ with the capped thresholds. For any value vector $\vec v \sim \dist$, either (1) the algorithms stopped at the same point having the same opening cost and value, or (2) $\alg$ stopped earlier at a threshold $\tau > n/\e$, while  $\alg'$ continued. In the latter case, the value $v$ that $\alg$ gets is greater than $n/\e$, while the value $v'$ that $\alg'$ gets is smaller, $v' \le v$. For such a scenario, the opening cost $c$ of $\alg$, and the opening cost $c'$ of $\alg'$ satisfy $c' \le c + n$. Thus, the total cost is $c' + v' \le c + v + n \le (1+\e) (c+v)$
Overall, we get that 
\[ \E{\dist}{\alg'} \leq
\E{\dist}{\alg}(1+\e).\]
\end{proof}

\begin{proof}[Proof of Theorem~\ref{thm:learning}]
With $\poly(n,\e,\log(1/\delta))$ samples from $\dist$, we obtain an empirical distribution $\hat \dist$. 

From Lemma~\ref{lem:closeness}, we have that with probability at least $1-\delta \e / \log(1/\delta)$, the following holds
\begin{align}\label{eq:error_thres}
    \bigg|  \E{v \sim \hat D}{\alg(\pi,\tau) - \min_{b\in \boxes} v_b}  - \E{v \sim D}{\alg(\pi,\tau) - \min_{b\in \boxes} v_b} \bigg| \le \e
\end{align}
for any permutation $\pi$ and any 
vector of thresholds $\vec v \in \lp[0, \frac n \e\rp]^n$. This gives us that we can estimate the cost of a threshold policy accurately.

To compare with the set of all partially adaptive policies that may not take the form of a threshold policy, we consider the set of scenario aware policies (SA).
These are policies $\text{SA}(\pi)$  parameterized by a permutation $\pi$ of boxes and are forced to visit the boxes in that order. However, they are aware of all values in the boxes in advance and know precisely when to stop. These are unrealistic policies introduced in \cite{ChawGergTengTzamZhan2020} which serve as an upper bound to the set of all partially adaptive policies.

As shown in \cite{ChawGergTengTzamZhan2020} (Lemma 3.3), scenario-aware policies are also learnable from samples. With probability at least $1-\delta \e / \log(1/\delta)$, it holds that for any permutation $\pi$
\begin{align}\label{eq:error_sa}
    \bigg| & \E{v \sim \hat D}{SA(\pi) - \min_{b\in \boxes} v_b} - \E{v \sim D}{SA(\pi) - \min_{b\in \boxes} v_b} \bigg| \le \e.
\end{align}
The $\alpha$-approximation guarantees (with $a \approx 4.428$) of Algorithm~\ref{alg:weitzman_main} hold even against scenario aware policies as there is no restriction on how the partially-adaptive policy may choose to stop. So for the empirical distribution, we can compute a permutation $\hat \pi$ and thresholds $\hat \tau$ such that:
$$
\E{\hat D}{\alg(\hat \pi,\hat \tau)} \le \alpha \cdot \min_\pi \E{\hat D}{  SA(\pi) }
$$
Clipping the thresholds to obtain $\hat \tau' = \min\{ \hat \tau, n/\e \}$, and letting $\Delta = \E{v \sim \hat D}{\min_{b\in \boxes} v_b} - \E{v \sim D}{\min_{b\in \boxes} v_b}$, we have that:
\begin{align*}
&\E{D}{\alg(\hat \pi,\hat \tau')} 
\le \E{\hat D}{\alg(\hat \pi,\hat \tau')} - \Delta + \e \\
&\le (1+\e) \E{\hat D}{\alg(\hat \pi,\hat \tau)} + \Delta + \e / 4 \\
&\le (1+\e) \alpha \cdot \min_{\pi} \E{\hat D}{SA(\pi)} - \Delta + \e / 4 \\
&\le (1+\e) \alpha \cdot \min_{\pi} \E{D}{SA(\pi)} + O( \Delta + \e )
\end{align*}

By Markov's inequality, we have that $\Pr{}{ \E{v \sim \hat D}{\min_{b\in \boxes} v_b} \le (1+\e) \E{v \sim D}{\min_{b\in \boxes} v_b} } \ge \frac \e { 1+ \e} \ge \e/2$.

Thus, repeating the sampling process $\frac {O(\log 1/\delta)} {\e}$ times and picking the empirical distribution with minimum $\E{v \sim \hat D}{\min_{b\in \boxes} v_b}$ satisfies $\Delta \le \e \E{v \sim D}{\min_{b\in \boxes} v_b}$ with probability at least $1-\delta$ and simultaneously satisfies equations \eqref{eq:error_thres} and \eqref{eq:error_sa}.

This shows that $\E{D}{\alg(\hat \pi,\hat \tau')} \le (1+O(\e)) \alpha \cdot \min_{\pi} \E{D}{SA(\pi)}$ which completes the proof by rescaling $\e$ by a constant.

\end{proof}

\bibliography{reference}

\newcommand{\etalchar}[1]{$^{#1}$}
\begin{thebibliography}{CGMT21}

\bibitem[AG11]{AzarGamz2010}
Yossi Azar and Iftah Gamzu.
\newblock Ranking with submodular valuations.
\newblock In {\em Proceedings of the Twenty-Second Annual {ACM-SIAM} Symposium
  on Discrete Algorithms, {SODA} 2011, San Francisco, California, USA, January
  23-25, 2011}, pages 1070--1079, 2011.

\bibitem[AGY09]{AzaGamzIftaYinr2009}
Yossi Azar, Iftah Gamzu, and Xiaoxin Yin.
\newblock Multiple intents re-ranking.
\newblock In {\em Proceedings of the 41st Annual {ACM} Symposium on Theory of
  Computing, {STOC} 2009, Bethesda, MD, USA, May 31 - June 2, 2009}, pages
  669--678, 2009.

\bibitem[ASW16]{AdamSvirWard2016}
Marek Adamczyk, Maxim Sviridenko, and Justin Ward.
\newblock Submodular stochastic probing on matroids.
\newblock {\em Math. Oper. Res.}, 41(3):1022--1038, 2016.

\bibitem[BDP22]{BechDughPate2022}
Curtis Bechtel, Shaddin Dughmi, and Neel Patel.
\newblock Delegated pandora's box.
\newblock In David~M. Pennock, Ilya Segal, and Sven Seuken, editors, {\em {EC}
  '22: The 23rd {ACM} Conference on Economics and Computation, Boulder, CO,
  USA, July 11 - 15, 2022}, pages 666--693. {ACM}, 2022.

\bibitem[BFLL20]{BoodFuscLazoLeon2020}
Shant Boodaghians, Federico Fusco, Philip Lazos, and Stefano Leonardi.
\newblock Pandora's box problem with order constraints.
\newblock In P{\'{e}}ter Bir{\'{o}}, Jason~D. Hartline, Michael Ostrovsky, and
  Ariel~D. Procaccia, editors, {\em {EC} '20: The 21st {ACM} Conference on
  Economics and Computation, Virtual Event, Hungary, July 13-17, 2020}, pages
  439--458. {ACM}, 2020.

\bibitem[BGK10]{BansGuptRavi2010}
Nikhil Bansal, Anupam Gupta, and Ravishankar Krishnaswamy.
\newblock A constant factor approximation algorithm for generalized min-sum set
  cover.
\newblock In {\em Proceedings of the Twenty-First Annual {ACM-SIAM} Symposium
  on Discrete Algorithms, {SODA} 2010, Austin, Texas, USA, January 17-19,
  2010}, pages 1539--1545, 2010.

\bibitem[BK19]{BeyhKlei2019}
Hedyeh Beyhaghi and Robert Kleinberg.
\newblock Pandora's problem with nonobligatory inspection.
\newblock In Anna Karlin, Nicole Immorlica, and Ramesh Johari, editors, {\em
  Proceedings of the 2019 {ACM} Conference on Economics and Computation, {EC}
  2019, Phoenix, AZ, USA, June 24-28, 2019}, pages 131--132. {ACM}, 2019.

\bibitem[CFG{\etalchar{+}}00]{CharFagiGuruKleiRaghSaha2002}
Moses Charikar, Ronald Fagin, Venkatesan Guruswami, Jon~M. Kleinberg, Prabhakar
  Raghavan, and Amit Sahai.
\newblock Query strategies for priced information (extended abstract).
\newblock In {\em Proceedings of the Thirty-Second Annual {ACM} Symposium on
  Theory of Computing, May 21-23, 2000, Portland, OR, {USA}}, pages 582--591,
  2000.

\bibitem[CGMT21]{ChawGergMcmaTzam2021}
Shuchi Chawla, Evangelia Gergatsouli, Jeremy McMahan, and Christos Tzamos.
\newblock Approximating pandora's box with correlations.
\newblock {\em CoRR}, abs/2108.12976, 2021.

\bibitem[CGT{\etalchar{+}}20]{ChawGergTengTzamZhan2020}
Shuchi Chawla, Evangelia Gergatsouli, Yifeng Teng, Christos Tzamos, and Ruimin
  Zhang.
\newblock Pandora's box with correlations: Learning and approximation.
\newblock In Sandy Irani, editor, {\em 61st {IEEE} Annual Symposium on
  Foundations of Computer Science, {FOCS} 2020, Durham, NC, USA, November
  16-19, 2020}, pages 1214--1225. {IEEE}, 2020.

\bibitem[CHKK15]{ChenHassKarbKrau2015}
Yuxin Chen, S.~Hamed Hassani, Amin Karbasi, and Andreas Krause.
\newblock Sequential information maximization: When is greedy near-optimal?
\newblock In {\em Proceedings of The 28th Conference on Learning Theory, {COLT}
  2015, Paris, France, July 3-6, 2015}, pages 338--363, 2015.

\bibitem[CJK{\etalchar{+}}15]{ChenJavdKarbBagnSrinKrau2015}
Yuxin Chen, Shervin Javdani, Amin Karbasi, J.~Andrew Bagnell, Siddhartha~S.
  Srinivasa, and Andreas Krause.
\newblock Submodular surrogates for value of information.
\newblock In {\em Proceedings of the Twenty-Ninth {AAAI} Conference on
  Artificial Intelligence, January 25-30, 2015, Austin, Texas, {USA.}}, pages
  3511--3518, 2015.

\bibitem[Dov18]{Dova2018}
Laura Doval.
\newblock Whether or not to open pandora's box.
\newblock {\em J. Econ. Theory}, 175:127--158, 2018.

\bibitem[FLT02]{FeigUrieLovaTeta2002}
Uriel Feige, L{\'{a}}szl{\'{o}} Lov{\'{a}}sz, and Prasad Tetali.
\newblock Approximating min-sum set cover.
\newblock In {\em Approximation Algorithms for Combinatorial Optimization, 5th
  International Workshop, {APPROX} 2002, Rome, Italy, September 17-21, 2002,
  Proceedings}, pages 94--107, 2002.

\bibitem[FLT04]{FeigUrieLovaTeta2004}
Uriel Feige, L{\'a}szl{\'o} Lov{\'a}sz, and Prasad Tetali.
\newblock Approximating min sum set cover.
\newblock {\em Algorithmica}, 40(4):219--234, 2004.

\bibitem[GGM06]{GoelGuhaMuna2006}
Ashish Goel, Sudipto Guha, and Kamesh Munagala.
\newblock Asking the right questions: model-driven optimization using probes.
\newblock In {\em Proceedings of the Twenty-Fifth {ACM} {SIGACT-SIGMOD-SIGART}
  Symposium on Principles of Database Systems, June 26-28, 2006, Chicago,
  Illinois, {USA}}, pages 203--212, 2006.

\bibitem[GJSS19]{GuptJianSing2019}
Anupam Gupta, Haotian Jiang, Ziv Scully, and Sahil Singla.
\newblock The markovian price of information.
\newblock In {\em Integer Programming and Combinatorial Optimization - 20th
  International Conference, {IPCO} 2019, Ann Arbor, MI, USA, May 22-24, 2019,
  Proceedings}, pages 233--246, 2019.

\bibitem[GK01]{GuptKuma2001}
Anupam Gupta and Amit Kumar.
\newblock Sorting and selection with structured costs.
\newblock In {\em 42nd Annual Symposium on Foundations of Computer Science,
  {FOCS} 2001, 14-17 October 2001, Las Vegas, Nevada, {USA}}, pages 416--425,
  2001.

\bibitem[GN13]{GuptNaga2013}
Anupam Gupta and Viswanath Nagarajan.
\newblock A stochastic probing problem with applications.
\newblock In {\em Integer Programming and Combinatorial Optimization - 16th
  International Conference, {IPCO} 2013, Valpara{\'{\i}}so, Chile, March 18-20,
  2013. Proceedings}, pages 205--216, 2013.

\bibitem[GNS16]{GuptNagaSing2016}
Anupam Gupta, Viswanath Nagarajan, and Sahil Singla.
\newblock Algorithms and adaptivity gaps for stochastic probing.
\newblock In {\em Proceedings of the Twenty-Seventh Annual {ACM-SIAM} Symposium
  on Discrete Algorithms, {SODA} 2016, Arlington, VA, USA, January 10-12,
  2016}, pages 1731--1747, 2016.

\bibitem[GNS17]{GuptNagaSing2017}
Anupam Gupta, Viswanath Nagarajan, and Sahil Singla.
\newblock Adaptivity gaps for stochastic probing: Submodular and {XOS}
  functions.
\newblock In {\em Proceedings of the Twenty-Eighth Annual {ACM-SIAM} Symposium
  on Discrete Algorithms, {SODA} 2017, Barcelona, Spain, Hotel Porta Fira,
  January 16-19}, pages 1688--1702, 2017.

\bibitem[GT22]{GergTzam2022}
Evangelia Gergatsouli and Christos Tzamos.
\newblock Online learning for min sum set cover and pandora's box.
\newblock In Kamalika Chaudhuri, Stefanie Jegelka, Le~Song, Csaba
  Szepesv{\'{a}}ri, Gang Niu, and Sivan Sabato, editors, {\em International
  Conference on Machine Learning, {ICML} 2022, 17-23 July 2022, Baltimore,
  Maryland, {USA}}, volume 162 of {\em Proceedings of Machine Learning
  Research}, pages 7382--7403. {PMLR}, 2022.

\bibitem[ISVDZ14]{ImSvirZwaa2012}
Sungjin Im, Maxim Sviridenko, and Ruben Van Der~Zwaan.
\newblock Preemptive and non-preemptive generalized min sum set cover.
\newblock {\em Mathematical Programming}, 145(1-2):377--401, 2014.

\bibitem[Sin18]{Sing2018}
Sahil Singla.
\newblock The price of information in combinatorial optimization.
\newblock In {\em Proceedings of the Twenty-Ninth Annual {ACM-SIAM} Symposium
  on Discrete Algorithms, {SODA} 2018, New Orleans, LA, USA, January 7-10,
  2018}, pages 2523--2532, 2018.

\bibitem[SW11]{SkutWill2011}
Martin Skutella and David~P. Williamson.
\newblock A note on the generalized min-sum set cover problem.
\newblock {\em Oper. Res. Lett.}, 39(6):433--436, 2011.

\bibitem[Wei79]{Weit1979}
Martin~L Weitzman.
\newblock {Optimal Search for the Best Alternative}.
\newblock {\em Econometrica}, 47(3):641--654, May 1979.

\end{thebibliography}
\bibliographystyle{alpha}

\appendix
\section{Appendix}

\subsection{Proofs from Section~\ref{sec:main_algo}}\label{apn:main_algo}

\mainThm*
The tighter guarantee proof follows the steps of the proof in section~\ref{subsec:proof_algo} for the opening cost, but provides a tighter analysis for the values cost. 
\begin{proof}[Tight proof of Theorem~\ref{thm:main_vs_pa_partial}]
Denote by $\sigma_s$ the reservation value for scenario $s$ when it was covered by $\alg$ and by $\tree$ the set of boxes opened i.e. the steps taken by the algorithm. Then we can write the cost paid by the algorithm as follows

\begin{equation}\label{eq:alg_cost_tight}
\alg  = \frac{1}{|\scenarios|} \sum_{s\in \scenarios} \sigma_s  = \frac{1}{|\scenarios|}\sum_{p\in \tree} |A_t| \sigma_p .
\end{equation}

We use the same notation as section~\ref{subsec:proof_algo} which we repeat here for convenience. Consider any point $p$ in the $\alg$ histogram, and let $s$ be its corresponding scenario and $t$ be the time this scenario is covered.
\begin{itemize}
\item $R_t:$ set of uncovered scenarios at step $t$
    \item $A_t:$ set of scenarios that $\alg$ chooses to cover at step $t$
    \item $c^*$: the opening cost such that $\gamma|R_t|$ of the scenarios in $R_t$ have opening cost less than $c^*$
    \item $R_{\text{low}} = \{s\in R_t: c_s^\opt \leq c^*\}$ the set of these scenarios
    \item $v^*$: the value of scenarios in $R_{\text{low}}$ such that $b|R_{\text{low}}|$ of the scenarios have value less than $v^*$
    \item $L = \{s\in R_{\text{low}}: v_s^\opt\leq v^*\}$ the set of scenarios with value at most $v^*$
    \item $B_L$: set of boxes the optimal uses to cover the scenarios in $L$ of step $t$
\end{itemize}
 The split described in the definitions above is again shown in Figure~\ref{fig:t_v_split_apn}, and the constants $1> \beta,\gamma>0$ will be determined in the end of the proof.

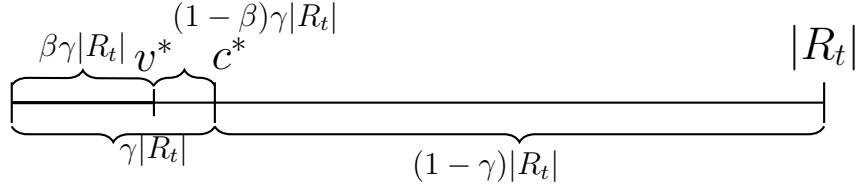
\begin{figure}[H]
				\centering
    \begin{tikzpicture}[scale=0.9]

\pgfmathtruncatemacro{\labelsY}{0.2};
\pgfmathtruncatemacro{\split}{3};

\draw[-,very thick, black] (0,0) -- (\split*0.7,0);
\draw[-,thick] (\split*0.7,0) -- (12,0);

\draw [decorate,black,thick,decoration={brace,amplitude=6pt}] (0,0.25) -- (\split*0.7,0.24) node[]{};
\draw [decorate,thick,decoration={brace,amplitude=6pt, mirror}] (0,-0.35) -- (\split,-0.35) node[]{};
\draw [decorate,thick,decoration={brace,amplitude=6pt, mirror}] (\split,-0.35)--(12,-0.35) node[]{};
\draw [decorate,thick,decoration={brace,amplitude=6pt}] (\split*0.7,0.25) -- (\split,0.25) node[]{};

\draw[-, thick] (0,0.3)-- (0,-0.3) node[]{};
\draw[-, thick] (12,0.3)-- (12,-0.3) node[label={[yshift=0.4cm]{ {\Large $|R_t|$}}}]{};
\draw[-, thick] (\split*0.7,0.2)-- (\split*0.7,-0.2) node[label={[yshift=0.3cm] {\Large $v^*$} }]{};
\draw[-, thick] (\split,0.3)-- (\split,-0.3) node[label={[xshift=0.2cm,yshift=0.42cm]{\Large $c^*$}}]{};

\draw (\split*0.7,-1.25) node[label={$\gamma|R_t|$}]{};
\draw (7,-1.5) node[label={$(1-\gamma)|R_t|$}]{};
\draw (\split*0.35,0.25) node[label={\textcolor{black}{$\beta \gamma|R_t|$}}]{};
\draw (1.2*\split,0.7) node[label={$(1-\beta)\gamma|R_t|$}]{};

\end{tikzpicture}
				\caption{Split of scenarios in $R_t$.}
    \label{fig:t_v_split_apn}
\end{figure}

Continuing from equation~\eqref{eq:alg_cost_tight} we obtain the following.
 \begin{align*}
    \alg & \leq \frac{1}{|\scenarios|}\sum_{t\in \tree} |A_t| \frac{|R_t|\sum_{b\in B_L} c_b 
        + \sum_{s\in L }v^\opt_s}{|L|} & \text{Inequality~\ref{eq:sigma_UB}}\\
    & \leq \frac{1}{|\scenarios|} \sum_{t\in\tree} \lp( |A_t|\frac{c^*}{\beta \gamma} 
        + \frac{\sum_{s\in L}v^\opt_s}{|L|}\rp) 
        & \text{Ineq.~\ref{eq:sigma_UB} and } |L| = \gamma \beta|R_t| \\
     & \leq \frac{\opt_o}{\beta \gamma(1-\gamma)} \sum_{t\in \tree}\frac{|A_t|}{|\scenarios|} + \sum_{t\in\tree}\frac{|A_t|}{|\scenarios|} \frac{\sum_{s\in L}v^\opt_s}{|L|} & \text{Since }c^* \leq \opt_o/(1-\gamma)\\
     & = \frac{\opt_o}{\beta \gamma(1-\gamma)} + \sum_{p\in\tree}\frac{|A_t|}{|\scenarios|} \frac{\sum_{s\in L}v^\opt_s}{|L|} 
             & \text{Since} \sum_{t}|A_t| = |\scenarios|
\end{align*}
                     
Where in the second to last inequality we used the same histogram argument from section~\ref{subsec:proof_algo}, to bound $c^*$ by $\opt_o/(1-\gamma)$. 

To bound the values term, observe that if we sorted the optimal values $v^\opt_s$ that cover each scenario by decreasing order, and denote $j_s$ the index of $v_s^\opt$ in this ordering, we add $v_s^\opt$ multiplied by the length of the interval every time $j_s\in \big[(1-\beta)\gamma|R_t|, \gamma |R_t|\big]$. This implies that the length of the intervals we sum up for $v_s^\opt$ ranges from $j_s/\gamma$ to $j_s/((1-\beta)\gamma)$, therefore the factor for each $v_s^\opt$ is
%
%
\[
 \frac{1}{\gamma} \sum_{i=j_s/\gamma}^{j_s/(1-\beta)\gamma } \frac{1}{i} 
\leq \frac{1}{\gamma} \log\lp( \frac{1}{1-\beta} \rp)\]

    We want to balance the terms $1/(\beta \gamma(1-\gamma))$ and $1/\gamma\log(1/(1-\beta))$ which gives that \[
    \gamma= 1- \frac{1}{\beta \log\lp( \frac{1}{1-\beta}\rp)}.\] 
    Since we balanced the opening cost and value terms, by substituting the expression for $\gamma$ we get that the approximation factor is 
    \[ 
    \frac{1}{\beta \gamma(1-\gamma)} = \frac{\beta \log^2\lp( \frac{1}{1-\beta}\rp)}{\beta\log\lp( \frac{1}{1-\beta}\rp) -1}.
    \]
    Numerically minimizing that ratio for $\beta$ and ensuring that $0<\beta,\gamma<1$ we get that the minimum is $4.428$ obtained at $\beta \approx 0.91$ and $\gamma\approx 0.55$.
\end{proof}

\subsection{Proofs from Section~\ref{subsec:equal}}\label{apn:equal}

\mainThmVarEqual*
\begin{proof}[Continued proof of Theorem~\ref{thm:main_vs_pa_equal}]
We now proceed to give the bound on the weights of the nodes of $\tree_\alg$. Consider any node $u$. We have that the weights at this node are equal to \[\s_u = \frac{c_{b_u} |R_u| + \sum_{s\in A_t}v^s_{b_{u}}}{|A_t|} 
		\leq \frac{c_b|R_u| + \sum_{s\in A}v^s_{b}}{|A|}\]
where the last inequality holds for all $A\subseteq R_u$ and any $b\in \boxes$.

Let $c^*_u$ the opening cost such that $\gamma|R_u|$ of the scenarios in $R_u$ have opening cost less than $c^*_u$, and by $R_{\text{low}} = \{s\in R_u: c_s^\opt \leq c^*_u\}$ the set of these scenarios. Similarly denote by $v^*_u$ the value of scenarios in $R_{\text{low}}$ such that $\beta |R_{\text{low}}|$ of the scenarios have value less than $v^*_u$ and by $L = \{s\in R_{\text{low}}^p: v_s^\opt\leq v^*_u\}$ these scenarios. This split is shown in Figure~\ref{fig:t_v_split}.

Note that, $c^*_u$ corresponds to the weights of node $u$ in $\tree^{(1-\gamma)}_{\opt_o}$, while the weights of node $u$ at $\tree^{(1-\gamma)}_{\opt_v}$ are at least $v^*_u$.

Let $B_L$ be the set of boxes that the optimal solution uses to cover the scenarios in $L$. Let $L_b \subseteq L \subseteq R_u$ be the subset of scenarios in $L$ that choose the value at box $b$ in $\opt$. Using inequality~\eqref{eq:sigma_p} with $b\in B_L$ and $A = L_b$, we obtain $\sigma_u |L_b| \leq c_b |R_u| + \sum_{s \in L_b} v_s^\opt $, and by summing up the inequalities for all $b\in B_L$ we get 
\begin{equation}\label{eq:sigma_UB}
\sigma_u \leq \frac{|R_u| \sum_{b\in B_L} c_b + \sum_{s\in L} v_s^\opt}{|L|}
		\leq \frac{|R_u| c^* + \sum_{s\in L} v_s^\opt}{|L|} 
		\leq \frac{c^*_u}{\beta \cdot \gamma} + v^*_u
\end{equation}
where for the second inequality we used that the cost for covering the scenarios in $L$ is at most $c^*_u$ by construction, and in the last inequality that $|L| = |R_t|/(\beta \cdot \gamma)$. We consider each term above separately, to show that the point $p$ is within the histograms.
\end{proof}
\genHist*

\begin{proof}[Proof of Lemma~\ref{lem:general_histogram}]

We denote by $\tree_u$ the subtree rooted at $u$, by $W(\tree) = \{ w: w \in \vec w_v \text{ for } v \in \tree \}$ the (multi)set of weights in the tree $\tree$. 
%
%
Denote, by $q^\rho( \tree )$ be the top $\rho$ percentile of all the weights in $\tree$. Finally, we define $Q(\rho | \tree)$ for any tree $\tree$ as follows:
\begin{itemize}
\item We create a histogram $H(x)$ of the weights in $W(\tree)$ in increasing order.
\item We calculate the area enclosed within $(1-\rho) |W(\tree)|$ until $|W(\tree)|$:
$$Q\lp( \rho | \tree \rp) = \int_{(1-\rho)|W(\tree)|}^{|W(\tree)|} H(x) dx$$
This is approximately equal to the sum of all the values greater than $q^\rho( \tree )$ with values exactly $q^\rho( \tree )$ taken fractionally so that exactly $\rho$ fraction of values are selected.
\end{itemize}

We show by induction that for every node $u$, it holds that $\rho \cdot \text{cost}(\tree^{(\rho)}_u) \le Q\lp( \rho | \tree \rp)$
\begin{itemize}
\item For the base case, for all leaves $u$, the subtree $\tree_u$ only has one node and the lemma holds as $\rho q^\rho(\tree_u) \le Q\lp( \rho | \tree_u \rp)$. 
\item Now, let $r$ be any node of the tree, and denote by $\child(r)$ the set of the children nodes of $r$.
\begin{align*}
    \rho \cdot \text{cost}(\tree^{(\rho)}_r) & = \rho \cdot
    q^\rho(\tree_r) |\vec w_r| + \rho  \cdot \sum_{v \in \child(r)} \text{cost}(\tree^{(\rho)}_v) & \text{Definition of cost($\tree^{(\rho)}_r$)} \\
    & \leq \rho \cdot q^\rho(\tree_r) |\vec w_r|  +  \rho \cdot \sum_{v\in \child(r)} Q(\rho| T_v)  
        & \text{From induction hypothesis}\\
        & \leq \rho \cdot q^\rho(\tree_r) |\vec w_r|  + Q\lp(\rho \frac {|W(\tree_r)| - |\vec w_r|} {|W(\tree_r)|}\, \Biggr| \, T_r\rp)& \text{Since }\tree_v \subseteq T_r\\
    & \leq Q\lp( \rho | T_r \rp)& 
    \end{align*}
    \end{itemize}
    
    The second-to-last inequality follows since $Q$ is defined as the area of the largest weights of the histogram. Including more weights only increases and keeping the length of the integration range the same (equal to $\rho (|W(\tree_r)| - |\vec w_r|)$) can only increase the value $Q$. 

    The last inequality follows by noting that if $H(x)$ is the histogram corresponding to the values of $\tree_r$, then
    \begin{align*}
    Q\lp( \rho | T_r \rp) - Q\lp( \rho \frac {|W(\tree_r)| - |\vec w_r|} {|W(\tree_r)|}\, \Biggr| \, T_r \rp) &= \int_{(1-\rho)|W(\tree_r)|}^{|W(\tree_r)|} H(x) dx 
    -
    \int_{(1-\rho)|W(\tree_r)| + \rho |\vec w _r|}^{|W(\tree_r)|} H(x)dx \\
    &=\int_{(1-\rho)|W(\tree_r)|}^{(1-\rho)|W(\tree_r)| + \rho |\vec w _r|} H(x) dx 
    \ge 
    \int_{(1-\rho)|W(\tree_r)|}^{(1-\rho)|W(\tree_r)| + \rho |\vec w _r|} q^\rho(\tree_r) dx \\
    &= \rho q^\rho(\tree_r) |\vec w _r|
    \end{align*}
    where the inequality follows since
    $H(x) \ge q^\rho(\tree_r)$ for $x \ge (1-\rho)|W(\tree_r)|$ by the definition of $q^\rho(\tree_r)$ as the top-$r$ quantile of the weights in $\tree_r$.
    

    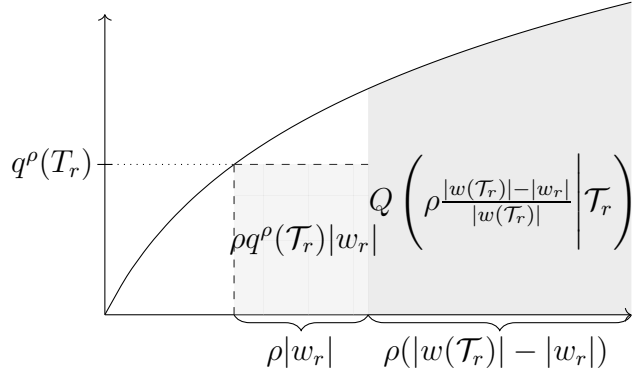
\begin{figure}[H]
				\centering
     \begin{tikzpicture}
\pgfmathtruncatemacro{\qpos}{2};
\pgfmathtruncatemacro{\totalX}{7};
\pgfmathtruncatemacro{\braceY}{-0.15};


    \fill [gray!30, opacity=.5, domain=3.5:\totalX, variable=\x]
      (3.5, 0)
      -- plot ({\x}, {2*ln (\x+1)})
      -- (\totalX, 0)
      -- cycle;

     \node[] at (\totalX-1.7, \qpos-0.5){$Q\lp( \rho \frac{|w(\tree_r)| - |w_{r}| }{|w(\tree_r)|}\Bigg| \tree_r \rp)$}; 
     
  \draw[domain=0:\totalX, smooth, variable=\x, black] plot ({\x}, {2*ln (\x+1)});

    \draw[->] (0,0)--(\totalX,0);

    \draw[->] (0,0)--(0,4);
    \node[] at (-0.7,\qpos) {$q^\rho(T_r)$};
    \draw[-] (-0.1,\qpos)--(0.1,\qpos);

    \draw[dotted] (0,\qpos)--(e-1,\qpos);
    \draw[dashed] (e-1,\qpos)--(3.5,\qpos);
    \draw[dashed] (e-1,2) -- (e-1,0);

    \fill[pattern={Lines[angle=45, distance=6mm, line width=13mm, xshift=5mm]}, pattern color=gray!30,opacity=.5] (e-1,0)--(e-1,\qpos)--(3.5,\qpos)--(3.5,0);    
    \node[] at (\totalX-4.35,\qpos-1) {$\rho q^\rho(\tree_r)|w_r|$};
    
    \draw [decorate,decoration={brace,amplitude=6pt, mirror}] (e-1,\braceY) -- (3.5,\braceY) node[]{};
    \node[] at (2.6,-0.5) {$\rho |w_r|$};

    \draw [decorate,decoration={brace,amplitude=6pt, mirror}] (3.5,\braceY) -- (\totalX,\braceY) node[]{};
    
    \node[] at (5.2,-0.5) {$\rho (|w(\tree_r)| - |w_r|)$};
\end{tikzpicture}
				\caption{Picture depicting the proof above.}
    \label{fig:histogram_apn}
\end{figure}
\end{proof}

\subsection{Proofs from Section~\ref{sec:learning}}\label{apn:learning}

\closeness*
\begin{proof}[Proof of Lemma~\ref{lem:closeness}]

We first argue that we can accurately estimate the cost for any vector of thresholds $\vec \tau$ when the order of visiting boxes is fixed. 

Consider any fixed permutation $\pi= \pi_1, \pi_2, \ldots , \pi_n$ be any permutation of the boxes, we relabel the boxes wlog so that $\pi_i$ is box $i$.

Denote by $\hat{V}_i = \min_{j\leq i} v_j$, and observe that $\hat{V}_i$ is a random variable that depends on the distribution $\dist$.  Then we can write the expected cost of the algorithm as the expected sum of the opening cost and the 
chosen value: $\E{\dist}{\alg} = \E{\dist}{\alg_o} + \E{\dist}{\alg_v}$. We have that:
\begin{align*}
\E{\dist}{\alg_o}  =\sum_{i=1}^n \Pr{\dist}{\text{reach }i} = \sum_{i=1}^n \Pr{\dist}{\bigwedge_{j=1}^{i-1} (\hat{V}_j > \tau_{j+1})}
\end{align*}
Moreover, we denote by $\overline{V}^i_{\vec \tau} = \bigwedge_{j=1}^{i-1} \lp( \hat{V}_j > \tau_{j+1} \rp)$ and we have
\begin{align*}
\E{\dist}{\alg_v - \hat{V}_n} & =\sum_{i=1}^n  \E{\dist}{ (\hat{V}_i - \hat{V}_n) \cdot \ind{\text{stop at }i}}\\
& = \sum_{i=1}^{n-1} \E{\dist}{(\hat{V}_i - \hat{V}_n) \cdot \ind{\overline{V}^i_{\vec \tau} \wedge \lp( \hat{V}_i \leq \tau_{i+1}\rp) }} \\
&  = \sum_{i=1}^{n-1} \mathbb{E}_{\dist} \Bigg[ \tau_{i+1} \Pr{r \sim U[0,\tau_{i+1}]}{r < \hat{V}_i - \hat{V}_n} \cdot \ind{ \overline{V}^i_{\vec \tau} \wedge \lp( \hat{V}_i \leq \tau_{i+1}\rp) }\Bigg]\\
&= \sum_{i=1}^{n-1} \tau_{i+1}  \textbf{Pr}_{\dist, r \sim U[0,\tau_{i+1}]}\Bigg[  {\overline{V}^i_{\vec \tau} \wedge \lp( r + \hat{V}_n \le \hat{V}_i \leq \tau_{i+1}\rp) }\Bigg]
\end{align*}

In order to show our result, we use from~\cite{BlumEhreHausWarm1989} that for a class with VC dimension $d<\infty$ that we can learn it with error at most $\e$ with probability $1-\delta$ using  $m=\text{poly}( 1/\e, d, \log\lp(  1/\delta\rp))$ samples.

Consider the class $\mathcal{F}_{\vec \tau}(\hat{V}, r) = {\bigwedge_{j=1}^{i-1} (\hat{V}_j > \tau_{j+1})} $. This defines an axis parallel rectangle in $\mathbb{R}^i$, therefore its VC-dimension is $2i$. Using the observation above we have that using $m=\text{poly}( 1/\e, n, \log\lp(  1/\delta\rp))$ samples, , with probability at least $1-\delta$, it holds
\[
\bigg|\Pr{\dist}{\mathcal{F}_{\vec \tau}(\hat{V}, r)}  - \Pr{\hat{\dist}}{\mathcal{F}_{\vec \tau}(\hat{V}, r)} \bigg| \leq \e \]
for all $\vec \tau \in\mathbb{R}^n$.
%

Similarly, the class $\mathcal{C}_{\vec \tau}(\hat{V}, r)=\bigwedge_{j=1}^{i-1} \lp( \hat{V}_j > \tau_{j+1} \rp) \wedge \lp( r + \hat{V}_n \le \hat{V}_i \leq \tau_{i+1}\rp)$ has VC-dimension $O(n)$ since it is an intersection of at most $n$ (sparse) halfspaces. Therefore, the same argument as before applies and for $m=\text{poly}( 1/\e, n, \log\lp(  1/\delta\rp))$ samples, we get 
%
\begin{align*}
\bigg|  \Pr{\dist, r \sim U[0,\tau_{i+1}]}{  \mathcal{C}_{\vec \tau}(\hat{V},r)}  - \Pr{\hat{\dist}, r \sim U[0,\tau_{i+1}]}{  \mathcal{C}_{\vec \tau}(\hat{V},r)} \bigg| \leq \e 
\end{align*}

for all $\vec \tau \in\mathbb{R}^n$, with probability at least $1-\delta$.

Putting it all together, the error can still be unbounded if the thresholds $\tau$ are too large. However, since we assume that $\tau_i \leq n/\e$ for all $i\in [n]$, $\poly(n, 1/\e, \log(1/\delta))$ samples suffice to get $\e$ error overall, by setting $\e \leftarrow \frac {\e^2}{n}$.

While we obtain the result for a fixed permutation, we can directly obtain the result for all $n!$ permutations through a union bound. Setting $\delta \leftarrow \frac \delta {n!}$ only introduces an additional factor of $\log(n!) = n \log n$ in the overall sample complexity.
\end{proof}
\end{document}